\documentclass[10pt,journal,twocolumn]{IEEEtran}
\newif\ifCLASSOPTIONromanappendices \CLASSOPTIONromanappendicestrue
 \usepackage{bm}
\usepackage{stfloats}
\usepackage{relsize}
\usepackage{amsthm}
\newtheorem{lemma}{Result}
\usepackage[utf8x]{inputenc}
\usepackage[T1]{fontenc}
\usepackage{url}
\usepackage{xfrac}
\usepackage{multicol}
\usepackage{xcolor}
\usepackage{ifthen}
\usepackage{amsfonts}
\usepackage{amssymb}
\usepackage{amsmath,hyperref}
\usepackage{sansmath}
\usepackage{siunitx}
\usepackage{mathtools}
\usepackage{nccmath}
\usepackage{cite}
\usepackage{systeme}
\usepackage{spalign}
\usepackage{caption}
\usepackage{algpseudocode}
\usepackage{algorithm}
\algnewcommand\algorithmicforeach{\textbf{for each}}
\algdef{S}[FOR]{ForEach}[1]{\algorithmicforeach\ #1\ \algorithmicdo}

\usepackage{eqparbox}
\newdimen{\algindent}
\algnewcommand\LeftComment[2]{%
\hspace{#1\algindent}$\triangleright$ \eqparbox{COMMENT}{#2} \hfill %
}
\algnewcommand\LeftCommentNoTriangle[2]{%
\hspace{#1\algindent} \eqparbox{COMMENT}{#2} \hfill %
}
\algnewcommand\LeftCommentNoIntent[1]{%
$\triangleright$ \eqparbox{COMMENT}{#1} \hfill %
}
\usepackage{verbatim}
\usepackage{tikz}
\usetikzlibrary{positioning,quotes,calc,fit,shapes,shadows,arrows}
\tikzset{block/.style={draw,very thick,text width=2cm,minimum height=4cm,align=center},
         line/.style={-latex}}
\tikzset{blockV/.style={draw,very thick,text width=2cm,minimum height=2cm, minimum width=4cm,align=center},
         line/.style={-latex}}
\tikzset{blockExt/.style={draw,very thick,minimum height=1cm, minimum width=1cm,align=center},
         line/.style={-latex}}
\usepgflibrary{arrows}
\usetikzlibrary{bayesnet}

\definecolor{light-gray}{HTML}{E0E0E0}

\makeatletter
\newcommand\notsotiny{\@setfontsize\notsotiny{6.82}{7.5}}
\makeatother
\newcommand{\B}[1]{\boldsymbol{#1}}
\newcommand\figscale{0.569}

\makeatletter
\newcommand{\labeltarget}[1]{\Hy@raisedlink{\hypertarget{#1}{}}}
\makeatother

\begin{document}
\title{Massive MIMO with Dense Arrays and 1-bit Data Converters}

\author{Amine Mezghani$^1$, \IEEEmembership{Member, IEEE}, Faouzi Bellili$^1$, \IEEEmembership{Member, IEEE}, and Robert W. Heath,~Jr.$^2$,  \IEEEmembership{Fellow, IEEE}
 \vspace{0.3cm}
\\\small $^1$E2-390 E.I.T.C, 75 Chancellor's Circle,  Winnipeg, MB, R3T 5V6, Canada.  \\
 $^2$890 Oval Drive
3100 Engineering Building II,
Raleigh, NC 27695, United States
  \vspace{0.1cm}
  \\\small Emails: 
  \{amine.mezghani,faouzi.bellili\}@umanitoba.ca, rwheathjr@ncsu.edu
  \vspace{0.3cm}
\thanks{A.~Mezghani and F.~Bellili are with the ECE Department at the University of Manitoba, Winnipeg, MB, Canada (emails:   \{Amine Mezghani, Faouzi.Bellili\}@umanitoba.ca). R.~W.~Heath is with the North Carolina State University (email: rwheathjr@ncsu.edu). This work was supported by Discovery Grants Program of the Natural Sciences and Engineering Research Council of Canada (NSERC) and the US National Science Foundation (NSF).}}

\maketitle
\begin{abstract}
We consider wireless communication systems with compact planar arrays having densely spaced antenna elements in conjunction with  one-bit analog-to-digital and digital-to-analog converters (ADCs/DACs). We provide closed-form expressions for the achievable rates with simple linear processing techniques for the uplink as well as the downlink scenarios while taking into account the effects of antenna mutual coupling. In the downlink case, we introduce the concept of non-radiating dithering to combat correlations of the quantization errors. Under higher antenna element density, we show that the performance of the quantized system can be made close to the ideal performance regardless of the operating signal-to-noise ratio.   
\end{abstract}



\section{Introduction}
The use of multi-antenna transmissions has shaped the wireless research and industry in the last two decades \cite{heath2018foundations}. With the advent of massive multiple-input-multiple-output (massive MIMO) \cite{Marzetta_2010}, high simultaneous multiplexing/diversity/array gains were made possible, thereby substantially improving energy and spectral efficiencies at the same time. Yet, several questions remain open: $i)$ the development of multi-standard super-wideband base stations, $ii)$ the realization of fully digital radio-frontend at millimeter-wave frequencies and above, and $iii)$ the convergence of the radio technology into a unified platform \cite{thomae_2019}. 


Much like the physical-layer wireless technology, the antenna technology has evolved immensely and the current trend is moving toward ultrathin compact/conformal antenna surfaces with subwavelength (i.e., microscopic) structures, also known as metasurfaces with 2D applied/induced magnetic and electric currents along the surface \cite{Renzo2020}. These new designs will enable multipurpose reconfigurable active/passive radiators/scatterers that can operate at several frequencies and enable more than one functionality at the same time via full 3D (i.e., $360^\circ$) beamforming. 

Conventional narrowband antenna arrays that are for instance still used in massive MIMO systems are generally based on a discrete array architecture where the element spacing is in the order of a half-wavelength. Newer antenna concepts, such as metasurfaces and tightly coupled dipole arrays \cite{Neto2006} are by contrast made of smaller elements with sub-wavelength spacing to synthesize a quasi-continuous aperture.

In the context of passive metasurfaces, the subwavelength structure offers a flexible and efficient wavefront manipulation such as beam-steering of the incident wave by optimizing the surface impedance distribution.
There are many other reasons for which subwavelength oversampling of the aperture can be advantageous also in the case of active arrays. First, for wideband antenna and ultra-wideband arrays with tightly coupled elements \cite{Jones:2007,Moulder:2012}, the inter-element spacing should be equal or smaller than the minimum operating half-wavelength to avoid grating lobes, leading to inherent oversampling at the lower operating frequencies. Second,  super-directivity \cite{williams2019communication} is possible with compact arrays albeit only at a small scale with an aperture size of few wavelengths due to the prohibitive losses and extreme narrowband behavior  \cite{Harrington:1960,Ivrlac:2010_2}. Third, near-field communication based on evanescent non-propagating fields (and thus not limited by the diffraction limit) can be envisioned for certain applications such as chip-to-chip communications and brain implant data transfer offering high multiplexing gains and high-security standard \cite{Phang2018}. 


Our work falls under the context of exploiting spatial oversampling to reduce the effects of nonlinear radio frontends, such as the quantization losses when low-resolution ADCs and DACs are used even down to just one-bit resolution per real dimension as a means to reduce implementation complexity and power consumption. This idea of trading simpler hardware for higher sampling in space has been already considered in prior work \cite{pirzadeh2020effect,Palguna2016,Scholnik2004,Krieger2013,Shao2019}. Indeed, the trend towards wideband arrays with subwavelength structures will make this coarse quantization concept even more attractive. Prior work, however, mainly focused on sigma-delta low-bit conversion that applies only for uniform linear arrays and necessitates sequential (i.e., non-parallelizable) processing.\\

\subsection{Contributions}
We derive bounds on the rate achievable with tightly coupled antenna arrays where the inter-element spacing is less than the half-wavelength for both the downlink as well as the uplink cases under 1-bit data converters. To this end, we derive a model for the antenna array response as well as for the spatial noise correlations that takes into account the antenna coupling effects based on the law of power conservation. The lower bound on the achievale rate is obtained through the Bussgang decomposition of the converters' output and by treating nonlinear distortion as addtive noise.   

For a fixed total antenna aperture, we show that, while the uplink and downlink system with infinite resolution ADCs and DACs do not benefit from the spatial oversampling, one-bit systems can substantially take advantage of the sub-wavelength sampling even with purely linear processing. To this end, we show the essential role of dithering applied prior to quantization to decorrelate the distortion error, a necessary feature to approach the ideal performance with linear processing. For the receiving arrays, dithering is performed naturally via the noise added by the low noise amplifiers (LNAs).  For transmitting arrays, this can be achieved by intentionally adding digital non-radiating (reactive) noise to the digital transmit signal before the one-bit DAC. 

Interestingly, since dithering can be fully controlled in the downlink case, one-bit transmitting  arrays appear to asymptotically approach the ideal performance at zero loss, and regardless of the effective SNR, as the number of spatial oversampling factor grows large. In the case of dense receiving arrays, a non-zero, but SNR independent, gap to the ideal case persists due to the fact that the dither signal from the LNAs is natural and cannot be optimized. This is in contrast to the prior work \cite{mezghani2007,mezghaniisit2007,mezghaniisit2008,mezghaniisit2009,koch,mezghani_itg_2010,koch1,hea14,Studer_2016,Jacobsson_2017,jianhua2014high,jacobsson2015one,juncil2015near,mollen2016performance_WSA,mollen2016performance,Yongzhisam,YongzhiGlobecom,Yongzhiuplink} focusing on uncoupled antenna arrays and suggesting the effectiveness of mid-rise one-bit sampling only in the low to medium SNR regime with a minimum of $2/\pi$ SNR loss.

Through numerical examples, we illustrate our findings regarding the performance of dense arrays with 1-bit ADCs and DACs. We also draw conclusions on the power-efficient implementation of tightly spaced antenna arrays with one-bit DAC.

\subsection{Notation}
Vectors and matrices are denoted by lower and upper case italic bold letters.  The operators $(\bullet)^\mathrm {T}$, $(\bullet)^\mathrm {H}$, $\textrm{tr}(\bullet)$ and $(\bullet)^*$ stand for transpose, Hermitian (conjugate transpose), trace, and complex conjugate, respectively.  The term ${\bf I}_M$ represents the identity matrix of size $M$. The vector $\boldsymbol{x}_i$ denotes the $i$-th column of a  matrix $\B{X}$ and $\left[\B{X}\right]_{i,j}$ denotes the ($i$th, $j$th) element, while $x_i$ is the $i$-th element of the vector $\B{x}$. We represent the Kronecker product of vectors and matrices by the operators  "$\otimes$". Additionally, $\textrm{Diag}(\boldsymbol{B})$ denotes a diagonal matrix containing only the diagonal elements of $\boldsymbol{B}$.
Further, we define $\B{C}_x={\rm E}[\B{x}\B{x}^{\rm H}] - {\rm E}[\B{x}]{\rm E}[\B{x}^{\rm H}]$ as the covariance matrix of $\B{x}$ and $\B{C}_{xy}$ as $\mathrm {E}[\B{x}\B{y}^{\rm H}]$.

\section{Dense Antenna Array Models}
The model for the dense antenna array used in this paper is in line with the prior work \cite{williams2019communication} and is based on power conservation arguments. We first derive a condition on any admissible model for densely spaced antenna architecture. The condition implies certain limitations on the achievable embedded radiation pattern of the array elements. We then argue that the widely used cosine based radiation pattern complies with this condition.

\subsection{Condition for power conservation at the far-field}

Consider a quasi-continuous infinite unidirectional antenna surface with element far-field effective area (i.e., embedded  pattern) $A_{e}(\theta, \varphi)$   for the elevation (aspect) and azimuth angles $\theta$ and $\varphi$ for all element while neglecting the edge effects. The far-field array steering vector of an $\sqrt{M} \times \sqrt{M}$ uniform linear array is
\begin{equation}
\begin{aligned}    
   & \boldsymbol{a}\left( \theta, \varphi,f \right)\\
& ~~    =\begin{bmatrix}  1 \\ e^{-2\pi {\rm j} \frac{a}{\lambda}  \sin \theta \sin \varphi } \\ \vdots  \\   e^{-2\pi {\rm j} \frac{a}{\lambda} (\sqrt{M}-1) \sin \theta \sin \varphi}  \end{bmatrix}  \otimes  \begin{bmatrix}  1 \\ e^{-2\pi {\rm j} \frac{a}{\lambda}  \sin \theta \cos \varphi } \\ \vdots  \\   e^{-2\pi {\rm j} \frac{a}{\lambda} (\sqrt{M}-1) \sin \theta \cos \varphi}  \end{bmatrix},
    \end{aligned}  
    \label{array_response}
\end{equation}
with $a$ being the element spacing and $\lambda$ is the wavelength. 

When dealing with spatial oversampling, i.e., $a \leq  \lambda/2$, it is crucial to  distinguish between extrinsic (environmental) and intrinsic (device) noise sources. At environmental  temperature $T$, the 
isotropic thermal radiation density for  one polarization is given by
\begin{equation}
N_{f}(f, T)=\frac{h f^{3}}{c^{2}} \frac{1}{e^{h f / N_0}-1} \stackrel{h f\ll N_0 }{\approx} \frac{N_0}{\lambda^2},  \quad \lambda=c_0/f 
\nonumber
\end{equation}
where the approximation  holds for the radio spectrum at ordinary temperatures. Consequently, the spatial covariance matrix associated with the environmental (i.e., extrinsic) noise is given by 
\begin{equation}
\begin{aligned}
\B{C}_{\rm n}(f)\!\!\!\!\!\!\!\!\!\!\!\!& 
\\
&=N_f(f,T) \!\!\int_0^{\frac{\pi}{2}}\!\!\!\! \int_{-\pi}^{\pi} \!\!   A_e(\theta,\varphi)  \boldsymbol{a} (\theta,\varphi,f) \boldsymbol{a} (\theta,\varphi,f)^{\rm H}   
{\rm sin}\theta {\rm d} \varphi {\rm d} \theta \\
&=N_0 \!\!\int_0^{\frac{\pi}{2}}\!\!\!\! \int_{-\pi}^{\pi} \!\! \frac{1}{\lambda^2}  A_e(\theta,\varphi)  \boldsymbol{a} (\theta,\varphi,f) \boldsymbol{a} (\theta,\varphi,f)^{\rm H}   
{\rm sin}\theta {\rm d} \varphi {\rm d} \theta \\
&\triangleq  N_0 \B{B}(f).
\end{aligned}
\label{matrix_B}
\end{equation}
We call the matrix $\B{B}(f)$ the coupling matrix.
Due to the Fourier structure of the array response in (\ref{array_response}), one can notice that the matrix $\B{B}(f)$ tend to be low rank  in case of  spatial oversampling with $a\ll \lambda$ and so does the extrinsic noise covariance matrix. This is due to the fact that $\boldsymbol{a} (\theta,\varphi,f)$  spans only a small portion of the $M$-dimensional space if $a\ll \lambda$. Based on the equipartition of energy principle, the average extrinsic energy associated with each complex degree of freedom received from the ambient environment  cannot exceed $N_0$, therefore we must have the property:
\begin{lemma}
Any time invariant antenna array fulfils $\B{B}(f) \preccurlyeq {\bf I} $.
\end{lemma}
\begin{proof}
The coupling matrix $\B{B}(f)$ is related to the antenna impedance $\B{Z}$ (or alternatively admittance matrix) by calculating the effective  covariance of the noise delivered by the the antenna system  to $M$ resistive loads with value $R_0$  as shown in Fig.~\ref{fig:ciruit} \cite{JAP1955,Ivrlac:2010_2} 
\begin{equation}
\begin{aligned}
\B{B}(f)=(R_0 {\bf I}+\B{Z}(f))^{-1} 4 \mathcal{R}\{\B{Z}(f)\}R_0 (R_0 {\bf I}+\B{Z}(f))^{-\rm H}.
\label{mat_frac}
\end{aligned}
\end{equation}
Then, we compare the denominator and numerator of the matrix fraction in (\ref{mat_frac})
\begin{equation}
\begin{aligned}
&(R_0 {\bf I}+\B{Z}(f))(R_0 {\bf I}+\B{Z}(f))^{\rm H}- 4 \mathcal{R}\{\B{Z}(f)\}R_0  \\
&\quad \quad \quad \quad  \quad  = (R_0 {\bf I}-\B{Z}(f))(R_0 {\bf I}-\B{Z}(f))^{\rm H} \curlyeqsucc
 {\bf 0},
\end{aligned}
\end{equation}
leading to the property $\B{B}(f)\preccurlyeq {\bf I}$.
\end{proof}
A possible admissible effective area that is compatible with the passivity condition  $\B{B}(f)\preccurlyeq {\bf I}$ is the cosine-shaped pattern (know as the normal gain with  uniformly illuminated aperture, c.f. Fig.~\ref{fig:cos_pattern})
\begin{equation}
 A_{e}(\theta, \varphi)=a^2 \cos \theta.
\end{equation}
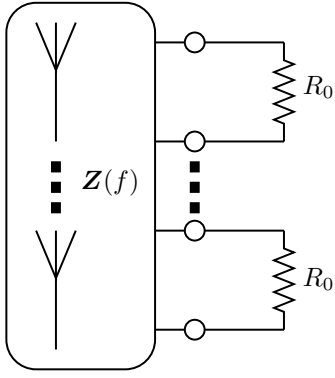
\begin{figure}
    \centering
\tikzset{every picture/.style={line width=0.75pt}} 
\begin{tikzpicture}[x=0.75pt,y=0.75pt,yscale=-1,xscale=1]
\draw   (90,40) -- (100,64) -- (110,40) (100,40) -- (100,100) ;
\draw   (165,100) .. controls (165,97.24) and (167.24,95) .. (170,95) .. controls (172.76,95) and (175,97.24) .. (175,100) .. controls (175,102.76) and (172.76,105) .. (170,105) .. controls (167.24,105) and (165,102.76) .. (165,100) -- cycle ;
\draw   (215,50) -- (215,59) -- (220,61) -- (210,65) -- (220,69) -- (210,73) -- (220,77) -- (210,81) -- (220,85) -- (210,89) -- (215,91) -- (215,100) ;
\draw   (215,145) -- (215,154) -- (220,156) -- (210,160) -- (220,164) -- (210,168) -- (220,172) -- (210,176) -- (220,180) -- (210,184) -- (215,186) -- (215,195) ;
\draw   (165,145) .. controls (165,142.24) and (167.24,140) .. (170,140) .. controls (172.76,140) and (175,142.24) .. (175,145) .. controls (175,147.76) and (172.76,150) .. (170,150) .. controls (167.24,150) and (165,147.76) .. (165,145) -- cycle ;
\draw   (165,195) .. controls (165,192.24) and (167.24,190) .. (170,190) .. controls (172.76,190) and (175,192.24) .. (175,195) .. controls (175,197.76) and (172.76,200) .. (170,200) .. controls (167.24,200) and (165,197.76) .. (165,195) -- cycle ;
\draw   (165,50) .. controls (165,47.24) and (167.24,45) .. (170,45) .. controls (172.76,45) and (175,47.24) .. (175,50) .. controls (175,52.76) and (172.76,55) .. (170,55) .. controls (167.24,55) and (165,52.76) .. (165,50) -- cycle ;
\draw    (175,50) -- (215,50) ;
\draw    (175,100) -- (215,100) ;
\draw    (175,145) -- (215,145) ;
\draw    (175,195) -- (215,195) ;
\draw   (90,145) -- (100,169) -- (110,145) (100,145) -- (100,205) ;
\draw [line width=3.75]  [dash pattern={on 4.22pt off 3.52pt}]  (170,110) -- (170,140) ;
\draw   (75,45) .. controls (75,36.72) and (81.72,30) .. (90,30) -- (135,30) .. controls (143.28,30) and (150,36.72) .. (150,45) -- (150,200) .. controls (150,208.28) and (143.28,215) .. (135,215) -- (90,215) .. controls (81.72,215) and (75,208.28) .. (75,200) -- cycle ;
\draw    (150,50) -- (165,50) ;
\draw    (150,100) -- (165,100) ;
\draw    (150,145) -- (165,145) ;
\draw    (150,195) -- (165,195) ;
\draw [line width=3.75]  [dash pattern={on 4.22pt off 3.52pt}]  (100,110) -- (100,140) ;
\draw (223,65.4) node [anchor=north west][inner sep=0.75pt]    {$R_{0}$};
\draw (223,162.4) node [anchor=north west][inner sep=0.75pt]    {$R_{0}$};
\draw (111,112.4) node [anchor=north west][inner sep=0.75pt]    {$\boldsymbol{Z}(f)$};
\end{tikzpicture}
    \caption{Circuit diagram of a MIMO antenna system with resistive loads}
    \label{fig:ciruit}
\end{figure}
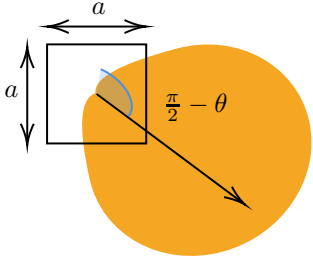
\begin{figure}
    \centering

\tikzset{every picture/.style={line width=0.75pt}} 

\begin{tikzpicture}[x=0.75pt,y=0.75pt,yscale=-1,xscale=1]

\draw  [color={rgb, 255:red, 245; green, 166; blue, 35 }  ,draw opacity=1 ][fill={rgb, 255:red, 245; green, 166; blue, 35 }  ,fill opacity=1 ] (119.52,61.59) .. controls (119.52,61.59) and (119.52,61.59) .. (119.52,61.59) .. controls (149.82,53.56) and (179.45,69.79) .. (185.7,97.85) .. controls (191.96,125.91) and (172.46,155.18) .. (142.16,163.21) .. controls (111.86,171.25) and (82.23,155.02) .. (75.98,126.95) .. controls (70.07,100.42) and (71.07,86.1) .. (79,84) .. controls (77.37,76.66) and (90.88,69.19) .. (119.52,61.59) -- cycle ;
\draw   (54,59) -- (104,59) -- (104,109) -- (54,109) -- cycle ;
\draw    (56,50) -- (102,50) ;
\draw [shift={(104,50)}, rotate = 180] [color={rgb, 255:red, 0; green, 0; blue, 0 }  ][line width=0.75]    (10.93,-3.29) .. controls (6.95,-1.4) and (3.31,-0.3) .. (0,0) .. controls (3.31,0.3) and (6.95,1.4) .. (10.93,3.29)   ;
\draw [shift={(54,50)}, rotate = 0] [color={rgb, 255:red, 0; green, 0; blue, 0 }  ][line width=0.75]    (10.93,-3.29) .. controls (6.95,-1.4) and (3.31,-0.3) .. (0,0) .. controls (3.31,0.3) and (6.95,1.4) .. (10.93,3.29)   ;
\draw    (44,108) -- (44,62) ;
\draw [shift={(44,60)}, rotate = 450] [color={rgb, 255:red, 0; green, 0; blue, 0 }  ][line width=0.75]    (10.93,-3.29) .. controls (6.95,-1.4) and (3.31,-0.3) .. (0,0) .. controls (3.31,0.3) and (6.95,1.4) .. (10.93,3.29)   ;
\draw [shift={(44,110)}, rotate = 270] [color={rgb, 255:red, 0; green, 0; blue, 0 }  ][line width=0.75]    (10.93,-3.29) .. controls (6.95,-1.4) and (3.31,-0.3) .. (0,0) .. controls (3.31,0.3) and (6.95,1.4) .. (10.93,3.29)   ;
\draw    (79,84) -- (152.4,138.8) ;
\draw [shift={(154,140)}, rotate = 216.75] [color={rgb, 255:red, 0; green, 0; blue, 0 }  ][line width=0.75]    (10.93,-3.29) .. controls (6.95,-1.4) and (3.31,-0.3) .. (0,0) .. controls (3.31,0.3) and (6.95,1.4) .. (10.93,3.29)   ;
\draw  [draw opacity=0][fill={rgb, 255:red, 74; green, 144; blue, 226 }  ,fill opacity=0.22 ] (81.09,71.52) .. controls (82.59,72.18) and (84.08,72.99) .. (85.55,73.95) .. controls (94.67,79.89) and (99.13,89.2) .. (95.51,94.76) .. controls (95.39,94.94) and (95.26,95.12) .. (95.12,95.29) -- (79,84) -- cycle ; \draw  [color={rgb, 255:red, 74; green, 144; blue, 226 }  ,draw opacity=1 ] (81.09,71.52) .. controls (82.59,72.18) and (84.08,72.99) .. (85.55,73.95) .. controls (94.67,79.89) and (99.13,89.2) .. (95.51,94.76) .. controls (95.39,94.94) and (95.26,95.12) .. (95.12,95.29) ;

\draw (111,82.4) node [anchor=north west][inner sep=0.75pt]    {$\frac{\pi }{2} -\theta $};
\draw (75,36.4) node [anchor=north west][inner sep=0.75pt]    {$a$};
\draw (31,78.4) node [anchor=north west][inner sep=0.75pt]    {$a$};

\end{tikzpicture}

    \caption{Unit element of the array: for a cosine shaped pattern $A_e=(\theta, \varphi)=a^2 \cos \theta$}
    \label{fig:cos_pattern}
\end{figure}
\begin{figure}
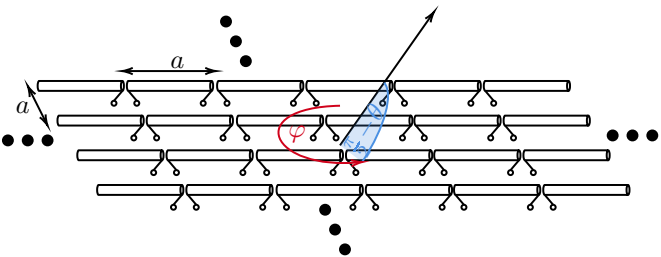

    \centering
\tikzset{every picture/.style={line width=0.75pt}} 

    \caption{Connected dipole array}
    \label{fig:connected_dipoles}
\end{figure}
A practical implementation using connected dipoles that approximately achieves this behavior can be found in \cite{Neto2006} and is illustrated in Fig.~\ref{fig:connected_dipoles}.   
\begin{lemma}
The cosine pattern  $A_e(\theta,\varphi)= a^2 \cos \theta $ fulfils the condition  $\B{B}(f) \preccurlyeq {\bf I} $.
\end{lemma}
\begin{proof}
Taking the case that $A_e(\theta,\varphi)= a^2 \cos \theta $, we prove that  $ \B{B}(f) \preccurlyeq {\bf I}$, which is equivalent to $ \B{f}^{\rm T}\B{B}(f) \B{f}^* \leq \|\B{f}\|^2_2,~ \forall \B{f}$. To this end,   we use the substitution 
\begin{align}
x&= \frac{a}{\lambda} \sin \theta \cos \varphi   \\
  y&=   \frac{a}{\lambda} \sin \theta \sin \varphi.
\end{align}
Then, we evaluate the quadratic form 
\begin{equation}
\begin{aligned}
 \B{f}^{\rm T}\B{B}(f) \B{f}^*& = 
 \int_0^{\frac{\pi}{2}} \!\!\int_{-\pi}^{\pi}   |\B{f}^{\rm T}  \boldsymbol{a} (\theta,\varphi,f)|^2 {\rm cos}\theta {\rm sin}\theta {\rm d} \varphi {\rm d}  \theta \\
&= \iint \limits_{x^2+y^2\leq  \frac{a^2}{\lambda^2}}  |\B{f}^{\rm T}  \tilde{\B{a}} (x,y,f)|^2 {\rm d} x {\rm d}   y   \\
&\leq   \int_{-\frac{1}{2}}^{\frac{1}{2}} \!\!\int_{-\frac{1}{2}}^{\frac{1}{2}}   |\B{f}^{\rm T}  \tilde{\B{a}} (x,y,f)|^2 {\rm d} x {\rm d}   y  \\
&=\|\B{f}\|^2_2 \quad \textrm{(Parseval's theorem).}
 \end{aligned}
 \label{proof_parseval}
\end{equation}
\end{proof}
Explicitly, for the case with $ A_e(\theta,\varphi) =a^2 \cos \theta $, the kernel matrix $\B{B}(f)$ is  obtained from the evaluation of the spherical  integral in (\ref{matrix_B}) as follows \cite{williams2019communication}
\begin{equation}
\begin{aligned}
&[\B{B}(f)]_{k+\sqrt{M}(\ell-1), m+\sqrt{M}(n-1) }  = \\
 & \begin{cases}
    \frac{\pi a^2}{\lambda^2}, & \text{for } k=m \text{ and } \ell=n \\
    \frac{a}{\lambda} \frac{ J_1(\frac{2\pi a}{\lambda}\sqrt{(k-m)^2+(\ell-n)^2})}{\sqrt{(k-m)^2+(\ell-n)^2}}, & \text{otherwise, }  
  \end{cases}
\end{aligned}
\end{equation}
where $J_1(\bullet)$ is the  Bessel function of first kind and first order. Finally, we also introduce the following quantity related to the diagonal entries of the matrix $\B{B}(f)$
\begin{equation}
\begin{aligned}
\gamma=\int_0^{\frac{\pi}{2}} \int_{-\pi}^{\pi} \frac{1}{a^2}  A_e(\theta,\varphi)  
{\rm sin}\theta {\rm d} \varphi {\rm d} \theta = \frac{\lambda^2}{a^2} [\B{B}(f)]_{j,j},
\end{aligned}
\end{equation}
which is equal $\pi$ for the cosine pattern. 
\subsection{Receive array and noise correlation}
Consider the case where the antenna array is operated in the receiving mode. The total effective noise is the sum of the contributions from  environmental  noise as characterized in (\ref{matrix_B}), and the intrinsic noise of the device that is dominated usually by the noise sources of the low noise amplifier (LNA). Denoting the noise figure (i.e., the noise enhancement factor) of the LNA by $N_F$, and assuming uncorrelated device noise, we obtain using (\ref{matrix_B})  the total noise covariance matrix 
\begin{equation}
\begin{aligned}
\B{C}_{\rm n,Tot}(f)\!\!\!\!\!\!\!\!\!\!\!\!& 
\\
&=\underbrace{N_0 \!\!\int_0^{\frac{\pi}{2}}\!\!\!\! \int_{-\pi}^{\pi} \!\! \frac{1}{\lambda^2}  A_e(\theta,\varphi)  \boldsymbol{a} (\theta,\varphi,f) \boldsymbol{a} (\theta,\varphi,f)^{\rm H}   
{\rm sin}\theta {\rm d} \varphi {\rm d} \theta}_{\rm extrinsic~noise} \\
&\hspace{5.5cm} +\underbrace{(N_F-1)N_0{\bf I}}_{\rm device~noise}, \\
&=  N_0 \B{B}(f)+ (N_F-1)N_0{\bf I},
\end{aligned}
\label{noise_covariance}
\end{equation}
where $N_F>1$ represents the noise figure of the receiver.  The assumption of uncorrelated device noise, although being not exact,is quite reasonable and provide meaningful prediction of the performance as shown later.  This uncorrelated intrinsic/device noise could also include the Ohmic losses of the antenna structure. Again, since $ \B{B}(f) \preccurlyeq {\bf I}$,  we have $ \B{C}_{\rm n,Tot}(f)  \preccurlyeq N_FN_0 {\bf I}$.



\subsection{Transmit array and radiated power}
When operating as a transmit array, we use the following  relationship between the element antenna pattern and its effective area which follows from thermodynamical arguments \cite{Rohlfs2004}
\begin{equation}
 G_e(\theta, \varphi)= \frac{4 \pi}{\lambda^2} A_e(\theta, \varphi). 
\end{equation}
When an excitation vector $\B{f}$ is applied to the antenna array, the total radiated power spectral density is obtained by integrating the radiated signal over a distance of  1 meter 
\begin{equation}
\begin{aligned}
\Phi_{\rm R,Tot}(f)
=& \frac{1}{4 \pi} \int_0^{\frac{\pi}{2}} \!\!\int_{-\pi}^{\pi}    G_e(\theta,\varphi) |\B{f}^{\rm T}  \boldsymbol{a} (\theta,\varphi,f)|^2
{\rm sin}\theta {\rm d} \varphi {\rm d} \theta  \\
=& \int_0^{\frac{\pi}{2}} \!\!\int_{-\pi}^{\pi}  \frac{1}{\lambda^2}  A_e(\theta,\varphi) |\B{f}^{\rm T}  \boldsymbol{a} (\theta,\varphi,f)|^2
{\rm sin}\theta {\rm d} \varphi {\rm d} \theta  \\
= &  \B{f}^{\rm T}  \B{B}(f) \B{f}^{*}\\
\leq& \|\B{f}\|_2^2,
\end{aligned}
\end{equation}
where $\B{B}(f)$ is defined in (\ref{matrix_B}) and the last step follows from $\B{B}(f) \preccurlyeq {\bf I}$. The radiated power is less or equal the available power $\|\B{f}\|_2^2$ and the equality holds only in the case $\B{B}(f)={\bf I}$, i.e., with perfect impedance matching conditions.  The radiation pattern $\B{f}^{\rm T}  \boldsymbol{a} (\theta,\varphi,f)$ is the truncated discrete-space Fourier transform (DTFT) of the beamforming vector $\B{f}$. 
We notice therefore the spatial smoothing and filtering property of dense antenna array, as only the lowest spatial frequencies $x^2+y^2\leq  \frac{a^2}{\lambda^2} \leq 1$ can propagate. Thus, spatial oversampling enables the reduction -- with a 2D cylindrical filter shape in the case of the cosine pattern (c.f. (\ref{proof_parseval}))--  of potential spurious radiations in the beamforming vector $\B{f}$ due for instance to hardware impairments. Since only the lower angular frequencies $x^2+y^2\leq  \frac{a^2}{\lambda^2} \leq 1$ can be radiated, any excitation outside this band will just generate reactive power that remains stored and confined very close (within few wavelengths) to the array surface.   Consequently, constant envelope transmission as well as low resolution D/A-converters can be used more appropriately in combination with densely spaced antenna arrays as suggested in prior work \cite{Krieger2013}.
\subsection{Far-field channel model}
Consider $K$ users in the far-field of the antenna array. The individual channels (assumed to be narrowband for simplicity) can be written  a superposition of $Q$ multi-path components 
\begin{equation}
\begin{aligned}
\B{h}_k &=\sum_{\ell=1}^Q s'_\ell \frac{1}{\lambda} \sqrt{A_e(\theta_\ell,\varphi_\ell)} \boldsymbol{a} (\theta_\ell,\varphi_\ell,f), 
\end{aligned}
\end{equation}
where $f=2\pi/\lambda$ is the carrier frequency.
In case of Rayleigh fading, i.e., rich scattering scenarios with isotropic angles of arrivals (for the uplink scenario) or departures (for the downlink scenario), the vector $\B{s}_k$ becomes independent and identically distributed (IID) Gaussian. In such a case, the channel covariance matrix is equal to $\B{B}(f)$ (c.f. (\ref{matrix_B}))
\begin{equation}
\begin{aligned}
{\rm E}[\B{h}_{k,\rm }\B{h}_{k,\rm }^{\rm H}] &=  \B{B}(f). 
\end{aligned}
\end{equation}
Accordingly, we can represent the channel vectors in the image space spanned by the matrix $\B{B}(f)$
\begin{equation}
\begin{aligned}
\B{h}_k &=  \B{B}(f)^{\frac{1}{2}} \B{s}_k, 
\end{aligned}
\end{equation}
with IID complex Gaussian vectors $\B{s}_k$. 
This shows again the significance of the coupling matrix $\B{B}(f)$ for the overall system model. 
\section{Spectral efficiency  with zero-forcing receiver and one-bit resolution ADCs}
We first consider the achievable rate infinite resolution ADCs and zero-forcing detection. The received signal is 
\begin{equation}
\begin{aligned}
\B{y} &= 
 \B{H}\B{x}+ \B{n},
\end{aligned}
\end{equation}
where $\B{n}$ is the effective noise vector having the covariance matrix given in (\ref{noise_covariance}), $\B{x}\in \mathbb{C}^K$ is the signal vector containing the transmitted signals from all the users with individual power values $\varepsilon_k$, and the channel matrix is defined as
\begin{equation}
\begin{aligned}
\B{H}=[\B{h}_1,\cdots,\B{h}_K].
\end{aligned}
\end{equation}
The individual achievable rates based on  zero-forcing equalizer  
\begin{equation}
\begin{aligned}
\B{G}=   (\B{H}^{\rm H}\B{H})^{-1} \B{H}^{\rm H},
\end{aligned}
\end{equation}
are given by  
\begin{equation}
\begin{aligned}
 R_k &= \log_2 \left(1+ \frac{{\varepsilon}_k}{[(\B{H}^{\rm H}\B{C}_{\rm n,Tot}^{-1}\B{H})^{-1}]_{kk}}\right)  \\
 &= \log_2 \left(1+ \frac{{\varepsilon}_k/N_0}{[(\B{S}^{\rm H}( {\bf I}+(N_F-1)\B{B}(f)^{-1}  \B{S})^{-1}]_{kk}}\right). 
\end{aligned}
\end{equation}
Since the matrix $\B{B}(f)$ is dominated by about $\frac{4a^2}{\lambda^2} M$ eigenvalues \cite{franceschetti_2017,williams2019communication} and $N_f$ is strictly larger than 1, the rate in a linear system mainly  depends on the effective aperture $\frac{4a^2}{\lambda^2} M$ rather than the number of antenna elements $M$ in dense arrays. In other words, half-wavelength spacing, i.e., $a=\lambda/2$, is nearly optimal in large antenna arrays with infinite resolution ADCs. This fact will be also apparent in the numerical results.  
\begin{figure}[h]
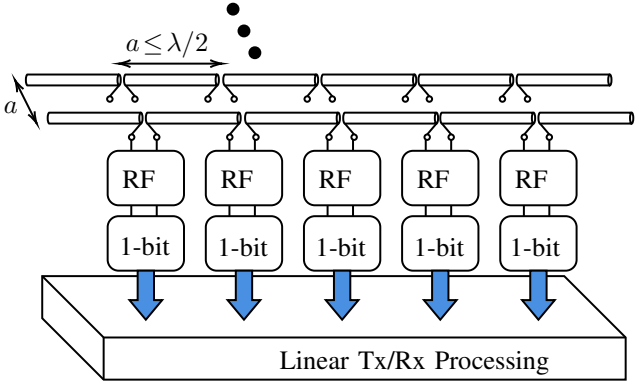

    \centering
\tikzset{every picture/.style={line width=0.75pt}} 

    \caption{Dense massive MIMO with 1-bit converters}
    \label{fig:1_bit_system}
\end{figure}


For the 1-bit case, shown in Fig.~\ref{fig:1_bit_system}, we resort to the Bussgang decomposition technique to derive a lower bound on the achievable rate. The rescaled quantized output is  
\begin{equation}
\B{y}_Q=  \frac{\B{D}^{\frac{1}{2}}}{\sqrt{2}} ( {\rm sign}(\mathcal{R}(\B{y}))+{\rm j} \cdot  {\rm sign}(\mathcal{I}(\B{y}))).
\end{equation}
where $\B{D}$ is a diagonal matrix containing the diagonal entries of $\B{C}_y$
\begin{equation}
\begin{aligned}
 \B{C}_y &=  \B{H} \B{P}\B{H}^{\rm H} +  \B{C}_{\rm n,Tot}, \quad \B{P}=\begin{bmatrix}
\varepsilon_1 & \cdots & 0 \\
 \vdots & \ddots &   \vdots \\
0 &  \cdots  & \varepsilon_K
\end{bmatrix}.
\end{aligned}
\end{equation}
The \emph{arcsine law} provides first the relationship between  covariance matrix  at the input  of the quantizer 
and the covariance matrix of quantized rescaled output
\begin{equation}
\begin{aligned}
\B{C}_{y_Q} &=\frac{2}{\pi}\B{D}^\frac{1}{2} {\rm arcsin} \left(\B{D}^{-\frac{1}{2}} \mathcal{R}(\B{C}_y) \B{D}^{-\frac{1}{2}}\right) \B{D}^\frac{1}{2} \\
&\quad \quad +{\rm j}\frac{2}{\pi}\B{D}^\frac{1}{2} {\rm arcsin} \left(\B{D}^{-\frac{1}{2}} \mathcal{I}(\B{C}_y) \B{D}^{-\frac{1}{2}}\right) \B{D}^\frac{1}{2},
\end{aligned}
\end{equation}
where the arcsine function is applied elementweise.
The effective noise covariance matrix can be obtained by subtracting the desired linear part of the Bussgang decomposition 
\begin{equation}
\begin{aligned}
\B{C}_{\rm n,Tot}' &=\B{C}_{y_Q} - \frac{2}{\pi}  \B{H}\B{P}\B{H}^{\rm H},
\end{aligned}
\end{equation}
which includes contributions from both the additive noise as well as from the quantization distortion. 

To further simplify the analysis, we resort to the following approximation under the assumption of small mutual correlations between the element  signals 
\begin{equation}
\begin{aligned}
\B{C}_{\rm n,Tot}' 
&\approx (1-\frac{2}{\pi}) \B{D} + \frac{2}{\pi}   \B{C}_{\rm n,Tot}.
\end{aligned}
\end{equation}
where we used the first order  approximation 
\begin{equation}
\begin{aligned}
\arcsin(x)  \approx  
  \begin{cases}
    \frac{\pi}{2}, & \text{for } x=1 \\
    x, & \text{for } |x| < 1.
  \end{cases}
\end{aligned}
\end{equation}
This approximation  leads to uncorrelated quantization noise (UQN) which is widely used in the literature. It holds however only  if the covariance matrix $\B{C}_y$ is dominantly diagonal. In other words, this approximation holds when the device noise dominates each individual element signal.

It is worth mentioning that the diagonal approximation assuming uncorrelated quantization noise might be unreasonable as shown in \cite{Mezghani2017} but can be made more realistic with sufficiently lower SNR per antenna (e.g. $\frac{a^2}{\lambda^2} \ll 1$). A lower bound on the mutual information  \cite{MezNos:Capacity-lower-bound:12} can be obtained by treating the distortion term as additive noise  
\begin{equation}
\begin{aligned}
 R_{k,\rm LB}^{\rm 1-bit}&= \log_2 \left( 1+ \frac{{\varepsilon}_k}{[(\frac{2}{\pi}\B{H}^{\rm H}\B{C}_{\rm n,Tot}'^{-1}\B{H})^{-1}]_{kk}}\right)    \\
 &\!\!\!\stackrel{\rm UQN}{\approx} \! \log_2 \!\! \left( 1+ \frac{{\varepsilon}_k}{[\B{H}^{\rm H}(\B{C}_{\rm n,Tot}+(\frac{\pi}{2}-1) \B{D} )^{-1}\B{H})^{-1}]_{kk}}\right) \!.    
\end{aligned}
\end{equation}
In the case of large number of isotropically distributed users,
we can make approximate the diagonal scaling matrix $\B{D}$, i.e., the diagonal entries of the covariance matrix $\B{C}_y$, as 
\begin{equation}
\begin{aligned}
\B{D} &\approx\left(  N_0+ \sum_k \varepsilon_k\right) {\rm Diag} \left(\B{B}(f)\right)+N_0(N_F-1) \\
&=\!\! \Bigg(\!(N_0\!\!+\!\!\!\sum_k \!\!{\varepsilon}_k) \frac{a^2}{\lambda^2}\!\! \underbrace{ \int_0^{\frac{\pi}{2}}\!\!\!\! \int_{-\pi}^{\pi} \!\! \!\! A_{e}(\theta,\varphi)
{\rm sin}\theta {\rm d} \varphi {\rm d} \theta}_{\triangleq \gamma ~(=\pi \rm ~ for ~the ~cosine ~ pattern) }\!+\!N_0(N\!_F\!-\!1)\!\!\Bigg)\!{\bf I} \\
&= \left((N_0 \!+\!  \sum_k {\varepsilon}_k)\frac{a^2}{\lambda^2} \gamma  + N_0 (N\!_F\!-\!1) \right) {\bf I}.
\end{aligned}
\end{equation}
We obtain therefore the approximation
\begin{equation}
\begin{aligned}
 &R_{k,\rm LB}^{\rm 1-bit} \stackrel{\rm UQN}{\approx} \log_2  \Bigg(1+\\
 &  \frac{{\varepsilon}_k/N_0}{[(\B{S}^{\rm H}\!( {\bf I}\!\! +\!\!(\frac{\pi}{2}(N\!\!_F\!-\!\!1)\!\!+\!\! ( \frac{\pi}{2} \!-\!1)\frac{a^2\gamma}{\lambda^2}(1\!\!+\!\!\sum\limits_k \! \frac{\varepsilon_k}{N_0}))\B{B}(\!f\!)^{\!-1})^{-1} \! \B{S})^{-1}]_{ \! kk}} \!\! \Bigg).
\end{aligned}
\end{equation}
Now we state the following asymptotic result.
\begin{lemma}
\label{result_uplink}
For infinitely dense arrays, the  SNR loss factor due the 1-bit quantization is 
\begin{equation}
\begin{aligned}
\lim_{a \longrightarrow 0} \frac{2^{R_{k,\rm LB}^{\rm 1-bit}}-1}{2^{R_{k}}-1} = \frac{N_F}{1+\frac{\pi}{2}(N\!_F\!-\!1)}.
\end{aligned}
\end{equation}
\end{lemma}
 Interestingly, the ratio can be made close to one if the noise figure of the receivers chains $N_F$ is close to one (i.e., nearly ideal receiver chain). Note that this result is independent of the SNRs $\varepsilon_k/N_0$  as opposed to the well-know $2/\pi$ loss factor  \cite{mezghani2007,MezNos:Capacity-lower-bound:12,Mezghani2017} with uncoupled antennas that requires operating at sufficiently low SNR.  
\section{Spectral efficiency  
with zero-forcing receiver and one-bit DACs}
Consider the ideal downlink case with infinite resolution DACs. The channel matrix can be represented as 
\begin{equation}
\begin{aligned}
\B{H}= \B{S}   \B{B}(f)^{\frac{1}{2}}. 
\end{aligned}
\end{equation}
At the base station, linear precoding with the matrix $\B{F} \in \mathbb{C}^{K \times M}$ is applied to the symbol vector $\B{x}$ of the users and the signal received at the users stations is
\begin{equation}
\begin{aligned}
\B{r}=\B{H}\B{F}\B{x} +\B{n}.
\end{aligned}
\end{equation}
As linear  transmit processing, we take the widely used zero-forcing precoder  which is given by 
\begin{equation}
\begin{aligned}
\B{F}= \sqrt{\frac{\varepsilon}{{\rm tr}((\B{H}\B{H}^{\rm H})^{-1})}} \B{H}^{\rm H} (\B{H}\B{H}^{\rm H})^{-1},
\end{aligned}
\end{equation}
where $\varepsilon$ represents the total available power for transmission. The actual  radiated power is however smaller than $\varepsilon$
\begin{equation}
\begin{aligned}
P_R =&  \frac{ {\rm tr} \left( \B{H}  \B{B}(f)   \B{H}^{\rm H}(\B{H}  \B{H}^{\rm H})^{-2} \right) }{{\rm tr}((\B{H} \B{H}^{\rm H})^{-1})}  \varepsilon   \\
= &\frac{ {\rm tr} \left( \B{S}  \B{B}(f)^2  \B{S}^{\rm H}(\B{S}\B{B}(f)  \B{S}^{\rm H})^{-2} \right) }{{\rm tr}((\B{S}\B{B}(f)  \B{S}^{\rm H})^{-1})}  \varepsilon \leq \varepsilon
\end{aligned}
\end{equation}
where the inequality is due to $\B{B}(f)\preccurlyeq {\bf I}$  and holds with equality for uncoupled antennas. In other words, only a portion of the available power $\varepsilon$ is radiated. Furthermore, the achievable rate under the zero-forcing precoder can be written as  
\begin{equation}
\begin{aligned}
 \bar{R}_k &= \log_2 \left(1+ \frac{\varepsilon/N_0/N_F}{{\rm tr}((\B{H}\B{H}^{\rm H})^{-1})}\right)  \\
 &= \log_2 \left(1+ \frac{\varepsilon/N_0/N_F}{{\rm tr}((\B{S}\B{B}(f) \B{S}^{\rm H})^{-1})}\right).
\end{aligned}
\label{R_ideal_down}
\end{equation}

The use of densely space antenna transmit arrays in combination with low resolution DACs and sigma-delta conversion  has been already considered in the literature, albeit only for the case of uniform linear arrays \cite{pirzadeh2020effect,Palguna2016,Scholnik2004,Krieger2013,Shao2019}. A generalization to  be applied for planar arrays is rather not straightforward due the sequential precoding. Instead, we will use the technique of dithering \cite{Saxena2019,amodh_2020} to whiten the quatization error by adding a random signal prior to the 1-bit quatizer. Furthermore, to avoid injecting interference in the environments, we propose the concept of non-radiating dithering, where the dither signal lies mainly on the null space of the coupling matrix $\B{B}(f)$ . To minimize the radiated power originating from the additive dithering signal, we consider a matrix $\B{U}$ describing the approximate null space \cite{Kokiopoulou_2007}   of the coupling matrix\footnote{$\B{B}(f)$ is an "almost" low-rank matrix for $a\ll \lambda$ with many small eigenvalues but not exactly low-rank. Therefore a matrix $\B{U}$ is constructed that is nearly orthogonal to $\B{B}(f)$.} $\B{B}(f)$, i.e.,
\begin{equation}
{\rm tr}(\B{U}^{\rm H}\B{B}(f) \B{U}) / {\rm tr}(\B{U}^{\rm H} \B{U})  \ll 1  .
\end{equation}
Then, the dithering is performed after linear precoding as follows
\begin{equation}
\begin{aligned}
\B{z} &= 
 \B{F}\B{x}+ \frac{\B{U}}{\|\B{U}\|_F}\B{v}_d
\end{aligned}
\label{dither_1}
\end{equation}
with Gaussian IID dithering vector $\B{v}_d$ having variance $\sigma_d^2$. 
A possible choice for $\B{U}$ is a projection on the approximate null space
\begin{equation}
\B{U}={\bf I}-(1+\delta) \B{B}(f)^{\frac{1}{2}}(\B{B}(f)+\delta {\bf I})^{-1}\B{B}(f)^{\frac{1}{2}},
\label{dither_2}
\end{equation}
for sufficiently small $\delta$ which is used for as a threshold for the eigenvalues of $\B{B}(f)$.
Subsequently,  the quantized  transmit vector is given by 
\begin{equation}
\B{z}_Q= \sqrt{\alpha}  \frac{\B{D}^{\frac{1}{2}}}{\sqrt{2}} ( {\rm sign}(\mathcal{R}(\B{z}))+{\rm j} \cdot  {\rm sign}(\mathcal{I}(\B{z}))) 
\end{equation}
where a scaling factor $\alpha$ is introduced to ensure the same radiated power $P_R$ as the ideal system. The rescaling matrix $\B{D}$ contains the diagonal elements of the covariance matrix of $\B{z}$
\begin{equation}
\begin{aligned}
\B{D}&= {\rm Diag}( \B{F}\B{F}^{\rm H}+  \frac{\B{U}\B{U}^{\rm H}}{\|\B{U}\|_F^2} \sigma^2_d)  \\
&=  {\rm Diag}( \B{F}\B{F}^{\rm H}+  \frac{\sigma^2_d}{M} {\bf I}) \\
&\approx \frac{\varepsilon+\sigma^2_d}{M} {\bf I},
\end{aligned}
\end{equation}
where the approximation holds for large number of users. With this approximation and assuming uncorrelated quantization noise, the radiated power of the 1-bit system has the following expression
\begin{equation}
\begin{aligned}
P_R^{\rm 1-bit} &=  \alpha \frac{2}{\pi} P_R +  \alpha {\rm tr} (\B{D}^{\frac{1}{2}} \B{R}_q \B{D}^{\frac{1}{2}} \B{B}(f) ) \\
&\approx  \alpha \frac{2}{\pi} P_R + \underbrace{\alpha (1-\frac{2}{\pi}) \frac{a^2}{\lambda^2} (    \varepsilon  +  \sigma^2_d) \gamma}_{\rm unwanted~isotropic~radiation} 
\end{aligned}
\end{equation}
 To allow for a fair comparison, we equate the radiated power of the 1-bit system to the ideal system, $P_R^{\rm 1-bit} = P_R $, yielding
\begin{equation}
\begin{aligned}
\alpha=\frac{P_R}{\frac{2}{\pi} P_R +  (1-\frac{2}{\pi}) \frac{a^2}{\lambda^2} (    \varepsilon  +  \sigma^2_d) \gamma}.
\end{aligned}
\end{equation}
Consequently, the noiseless received signal at the users terminals is 
\begin{equation}
\B{H}\B{z}_Q= \sqrt{\alpha\frac{2}{\pi}  \frac{\varepsilon}{{\rm tr}((\B{H}\B{H}^{\rm H})^{-1})}} \B{x}+ \sqrt{\alpha} \B{H}  \B{D}^{\frac{1}{2}} \B{q}.
\end{equation}
Finally, a lower bound on the achievable rate with 1-bit DACs can be obtained by simply modeling the quantization error as additive noise that is received through the channel  
\begin{equation}
\begin{aligned}
 \bar{R}_{k,\rm LB}^{\rm 1-bit} &= \log_2 \left(1+ \frac{\frac{2}{\pi}\alpha \varepsilon/{\rm tr}((\B{H}\B{H}^{\rm H})^{-1}) }{N_0N_F+  \alpha[\B{H}\B{D}^{\frac{1}{2}} \B{R}_q \B{D}^{\frac{1}{2}} \B{H}^{\rm H}]_{kk}}\right)  \\
 &\stackrel{\rm UQN}{\approx}  \log_2 \left(1+ \frac{\frac{2}{\pi}\alpha \varepsilon/{\rm tr}((\B{H}\B{H}^{\rm H})^{-1}) }{N_0N_F+  \alpha (1-\frac{2}{\pi}) \frac{a^2}{\lambda^2} (    \varepsilon  +  \sigma^2_d) \gamma  }\right) \\
  &=  \log_2 \left(1+ \frac{\frac{2}{\pi}\alpha \varepsilon/{\rm tr}((\B{H}\B{H}^{\rm H})^{-1}) }{N_0N_F+  (1-\frac{2}{\pi}\alpha )\varepsilon } \right) \\
   &=  \log_2\! \left(\! 1\! +\!  \frac{ \varepsilon/{\rm tr}((\B{H}\B{H}^{\rm H})^{-1})/(N_0N_F) }{1\! +\!  ( \frac{1}{P_R}\! +\! \frac{1}{N_0N_F}) (\frac{\pi}{2}\! -\! 1) \frac{a^2}{\lambda^2}  ( \varepsilon \!   + \!   \sigma^2_d) \gamma } \! \right). \\
\end{aligned}
\label{perf_1_bit_downlink}
\end{equation}
By comparing with the ideal case in (\ref{R_ideal_down}), we state the following result.
\begin{lemma}
For infinitely dense arrays, the  SNR loss factor due the 1-bit DACs vanishes under the same radiated power as the ideal system
\begin{equation}
\begin{aligned}
\lim_{a \longrightarrow 0} \frac{2^{\bar{R}_{k,\rm LB}^{\rm 1-bit}}-1}{2^{\bar{R}_{k}}-1} = 1.
\end{aligned}
\end{equation}
\end{lemma}
Interestingly, this results is better than Result~\ref{result_uplink} obtained for the uplink case with 1-bit ADCs. This can be explained by the fact that the additive dithering noise can be engineered in the DAC case, while the device noise for the ADC case is in general white and does not offer as much flexibility. It is worth mentioning that the choice of the dither variance $\sigma_d^2$ is critical. It has to be large enough to reduce the quantization noise correlation but excessive values would lead to a penalty as shown in the rate formula (\ref{perf_1_bit_downlink}). This trade-off and the corresponding optimization of $\sigma_d^2$ have been considered for the case of conventional antenna arrays in \cite{Saxena2019,amodh_2020}. In our setting, the optimal value for $\sigma_d^2$  will depend on the number of users, antennas, and antenna spacing. 
\section{Numerical examples}

We consider in this section some numerical examples and particular cases. In particular, we will discuss the validity of the common approximation of uncorrelated
distortion error and draw key properties of densely spaced antenna systems with low-resolution converters. In fact, the validity of this
approximation will also indicate how close the performance is to the ideal system when relying on simple linear processing \cite{Mezghani2017}.

In all numerical examples, we fix  the total geometric aperture to be  $a \sqrt{M} \times a \sqrt{M}=2.5 \lambda \times 2.5 \lambda $, while varying $M$ and $a$. Furthermore, we assume a cosine-based element pattern  $A_{e}(\theta, \varphi)=a^2 \cos \theta$ since it is physically feasible. In addition, we consider a  rich scattering environment with  $K=2$ users and isotropic multipath components, leading to a user's channel covariance matrix that is equal to ${\rm E}[\B{h}_k\B{h}_k^{\rm H}]=\B{B}(f)$. Accordingly, the performance results are averaged over 100 channel realizations.    
\subsection{Uplink case}
Consider fist the uplink case with $\varepsilon_k/N_0=2$ for both users and $N_F=2$. The corresponding achievable rates are shown in Fig.~\ref{fig:RX}  versus the number elements for the ideal case with infinite resolution as well as the one-bit case, both with and without the uncorrelated quantization noise assumption (UQN). The element count is increased from $M=25$ ($\lambda/2$ inter-element spacing) to $M=400$ ($\lambda/8$ inter-element spacing). The ideal system does not significantly take advantage of increasing the element density above half-wavelength spacing and the performance is mainly given by the total physical aperture. In fact, array gain beyond the normal geometric aperture of the antenna is known as super-gain but occurs practically only at small to moderate electrical size $\sqrt{2M} \frac{a}{\lambda}$ \cite{Harrington:1960} and requires, in general, a matching network \cite{Ivrlac:2010_2}. By contrast, the system with 1-bit ADCs clearly benefits from spatial oversampling and a substantial gain is obtained particularly when moving from   $\lambda/2$ inter-element spacing to $\lambda/4$ inter-element spacing. We also observe that the UQN model tends to overpredict the performance particularly for $M=25$, i.e., $a=\lambda/2$  but the discrepancy to the exact bound decreases as we increasingly overpopulate the antenna structure. This shows that ideal performance can be approached with  1-bit ADCs simple linear processing, without the need for complex spatial sigma-delta conversion as proposed in the prior work.      
\begin{figure}[h!]
\centering
\includegraphics[scale=\figscale]{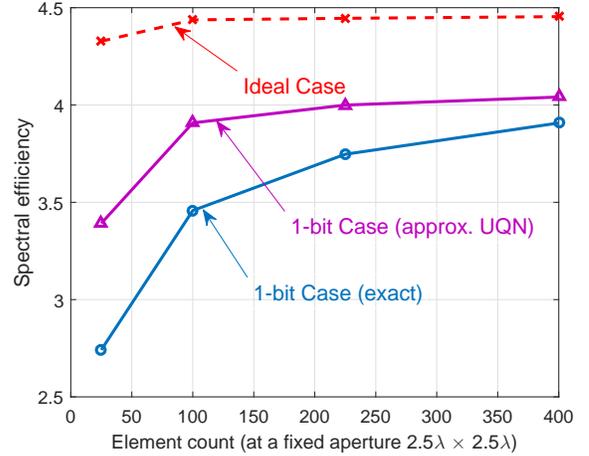}
\caption{Spectral efficiency versus element count for the uplink scenario at fixed aperture $2.5\lambda\times 2.5\lambda$ with $\varepsilon_k/N_0=2$, $N_F=2$, and two users.}
\label{fig:RX}
\end{figure}
\subsection{Downlink case}
In the downlink scenario, a quite similar system setting is considered with 2 users and a rich scattering environment. We fix the available power to be $\varepsilon/(N_0N_F)=2$ and the total aperture to be $2.5 \lambda \times 2.5 \lambda $ and vary the number of elements from 25 up to 400. Again, the achievable rate of the ideal system hardly changes as shown in Fig.~\ref{fig:TX} for similar reasons as for the uplink case. Furthermore, the 1-bit system performs poorly if no dithering is introduced (dashed line) and does hardly improve with higher antenna count. This is explained by the strong correlations among the individual distortion errors that combine coherently across the antenna/channel and severely affect the effective signal-to-noise-and-distortion ratio.  In such a case, it is important to avoid the risk of relying on the oversimplified UQN model that might lead to misleading insights. Nevertheless, by adding non-radiating dithering according to (\ref{dither_1}) and (\ref{dither_2}), the performance improves substantially and the UQN becomes more accurate as shown in Fig.~\ref{fig:TX}. To ensure this, the power of the dither signal $\sigma_d^2$ is increased proportionally to $\frac{\lambda}{a}$ to ensure the UQN property. Since this increase is still slower than $\frac{\lambda^2}{a^2}$, the performance derived in (\ref{perf_1_bit_downlink}) can still arbitrarily approach the ideal performance. Fig.~\ref{fig:TX_power} illustrates the ratio of the radiated power $P_R$ to  power of the excitation vector, being  $\varepsilon$ for the ideal system and  $\alpha(\varepsilon+\sigma_d^2)$ for the 1-bit system. The fact that the ratio for the 1-bit system reduces with the number of elements means that most of the power that circulates in the antenna is reactive in this regime. Since constant-envelope low-resolution signals are exciting the elements, this suggests the use of efficient switched-mode amplifiers exploiting the capability of the array for storing and recycling reactive power present in the evanescent field.   

\begin{figure}[h!]
\centering
\includegraphics[scale=\figscale]{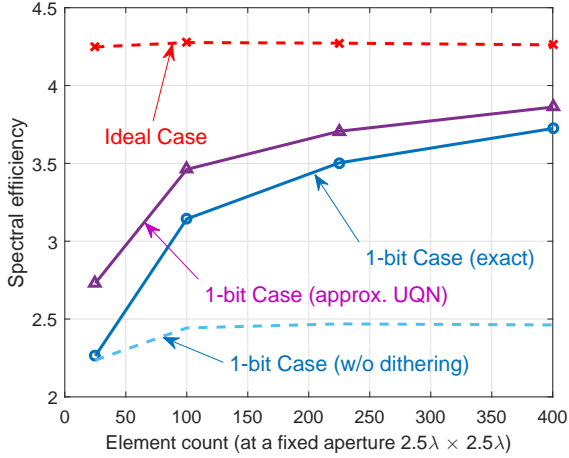}
\caption{Spectral efficiency versus element count for the downlink scenario at fixed aperture $2.5\lambda\times 2.5\lambda$ with $\varepsilon/(N_0N_F)=2$, $\frac{\sigma_d^2}{\varepsilon}=\frac{\lambda}{3a}$,  $\delta=0.01$ and two users.}
\label{fig:TX}
\end{figure}

\begin{figure}[h!]
\centering
\includegraphics[scale=\figscale]{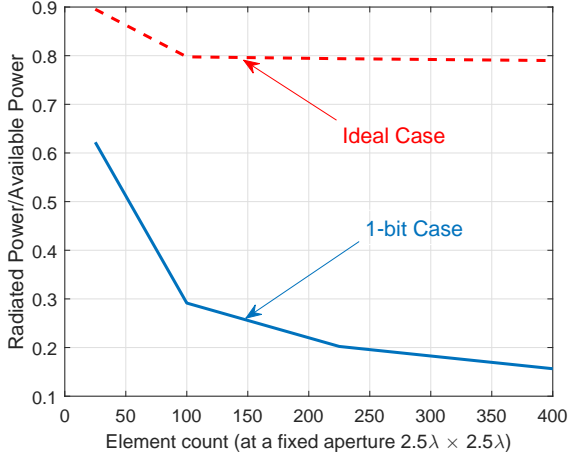}
\caption{Ratio $P_R/\varepsilon$ in the ideal system and $P_R/(\alpha\varepsilon+\alpha\sigma_d^2)$ in the 1-bit system versus element count for the downlink scenario at fixed aperture $2.5\lambda\times 2.5\lambda$ with $\varepsilon/(N_0N_F)=2$, $\frac{\sigma_d^2}{\varepsilon}=\frac{\lambda}{3a}$,  $\delta=0.01$ and two users.}
\label{fig:TX_power}
\end{figure}

\section{Conclusion}
Compact antenna arrays with inter-element spacing smaller than a half-wavelength are required for future radio systems to operate in different frequency bands and serve several applications and functionalities. We provided some requirements on any models used for this type of antennas to be physically consistent. Additionally, we showed that tightly spaced multiantenna systems can help to relax the linearity requirement of the radio frontend  by using 1-bit converters and exploiting their spatial filtering feature provided by nature and at no extra digital processing complexity. In particular, power amplifiers can be operated in saturated or switched-mode, while most of the generated distortion manifests as  evanescent-wave  
and does not radiate. This opens new directions for designing waveforms and processing techniques dedicated to this type of antenna structures.

\bibliographystyle{IEEEtran}
\bibliography{IEEEabrv,references.bib}
 
\end{document}